\documentclass[conference]{IEEEtran}
\IEEEoverridecommandlockouts

\usepackage{cite}
\usepackage{amsmath,amssymb,amsfonts}
\usepackage{graphicx}
\usepackage{textcomp} 
\usepackage{xcolor}
\usepackage{multirow}
\usepackage{pifont}
\usepackage{amsthm} 
\usepackage{enumitem} 
\usepackage{nicefrac}  
\usepackage{booktabs}   
\usepackage{bm}
\usepackage{subcaption}
\usepackage{algorithm}
\usepackage{algpseudocode}
\usepackage{hyperref}
\newtheorem{proposition}{proposition} 
\newtheorem{theorem}{theorem}

\def\BibTeX{{\rm B\kern-.05em{\sc i\kern-.025em b}\kern-.08em
    T\kern-.1667em\lower.7ex\hbox{E}\kern-.125emX}}

\newcommand{\Dset}{\mathcal{D}}
\newcommand{\Dpos}{\Dset^{+}}
\newcommand{\Dneg}{\Dset^{-}}
\newcommand{\Tpos}{\mathcal{T}^{+}}
\newcommand{\Tneg}{\mathcal{T}^{-}}
\newcommand{\TCP}{\mathrm{TCP}}
\newcommand{\TCPn}{\mathrm{TCP}_{\mathrm{n}}}
\newcommand{\TCPa}{\mathrm{TCP}_{\alpha}}
\newcommand{\MCP}{\mathrm{MCP}}
\newcommand{\ind}[1]{\mathbb{I}\!\left[#1\right]}

\newcommand{\prob}{\mathbb{P}}
\newcommand{\R}{\mathbb{R}}
\newcommand{\pstar}{p^{*}}
\newcommand{\phat}{\hat{p}}
\newcommand{\chat}{\hat{c}}
\newcommand{\yhat}{\hat{y}}
\newcommand{\ystar}{y^{*}}
\newcommand{\PE}{\mathrm{PE}}
\newcommand{\Margin}{\mathrm{Margin}}
\newcommand{\Energy}{\mathrm{E}}

\theoremstyle{plain}

\theoremstyle{remark}
\newtheorem{remark}{Remark}

\begin{document}

\title{ 
$TCP_\alpha$: Margin-Controlled Confidence estimation for reliable Music Information Retrieval
 \\
 

}

\author{
\IEEEauthorblockN{Parampreet Singh\IEEEauthorrefmark{1},
Anushka Singh\IEEEauthorrefmark{3},
Sumit Kumar\IEEEauthorrefmark{2}, 
Vipul Arora\IEEEauthorrefmark{4}}
\IEEEauthorblockA{
Indian Institute of Technology, Kanpur, India\\
Email: \IEEEauthorrefmark{1}params21@iitk.ac.in, 
\IEEEauthorrefmark{3}anushkas22@iitk.ac.in,
\IEEEauthorrefmark{2}krsumit@iitk.ac.in, 
\IEEEauthorrefmark{4}vipular@iitk.ac.in}
}

\maketitle

\begin{abstract}
Deep neural networks are often overconfident, assigning high confidence even to incorrect predictions. Consequently, users lack a reliable signal for deciding when a prediction can be trusted.
Post-hoc confidence estimation addresses this by training a lightweight auxiliary head over a frozen classifier.
Existing targets, however, suffer from inherent ambiguity: they assign overlapping confidence values to correct and incorrect predictions, while errors near the decision boundary receive confidence scores indistinguishable from correct predictions. In this work, we propose $TCP_\alpha$, a novel confidence target that resolves these limitations by introducing a margin-controlled penalty for misclassified samples. We prove that $TCP_\alpha$ guarantees complete separation between the target values of correct and incorrect predictions, with a separation margin that is independent of the number of classes and increases monotonically with the penalty parameter.
Since accurate classifiers naturally produce very few errors, learning these targets results in a severely imbalanced regression problem. We therefore present a systematic study of training strategies for learning under this imbalance and identify an effective training configuration through extensive ablation studies. We evaluate the proposed approach on r\=aga identification, investigate its robustness under domain shift, and further validate it on frame-wise ornamentation detection without modifying the selected configuration. Across all settings, $TCP_\alpha$ consistently outperforms existing confidence targets for failure prediction. Rejecting only the least-confident 8\% of predictions improves the base model's macro-F1 from 0.89 to 0.98, while fine-tuning the confidence head with only 5\% labeled samples from a new corpus effectively restores performance under domain shift.
It provides an effective and transferable framework for reliable failure prediction in practical music information retrieval systems.
\end{abstract}


\section{Introduction}
\IEEEPARstart{D}{eep} learning has significantly advanced Music Information Retrieval (MIR), enabling state-of-the-art performance across tasks such as automatic transcription~\cite{saha2026_transcription_muller,transcription_juhan}, source separation~\cite{sourceseparation,source_separation}, instrument recognition~\cite{instrument_TASLP}, genre classification~\cite{genre_2}, \textit{r\=aga} identification~\cite{Paramsingh2024explainabledeeplearninganalysis}, and rhythmic-pattern analysis~\cite{drums,kodag2025meta}. Despite these advances, most deep learning models provide only a class prediction without indicating whether that prediction should be trusted.
This absence of reliable confidence estimates limits their usefulness in applications such as music pedagogy~\cite{kumar2026automatic}, where an automated critique delivered with unwarranted certainty can mislead learners, and equally in recommendation systems, where an unreliable prediction shown with high confidence erodes user trust.
Labeled data is a standing bottleneck across MIR, and automatically labeling archives with tags such as \textit{r\=aga} or genre could ease the scarcity and cost of expert annotation~\cite{labelling_1}.
Such a system is useful, however, only if it can identify the predictions it is likely to get wrong and defer them for expert review.
Confidence estimation~\cite{corbiere2019_confidence_main_model,corbiere2021confidence_confidnet_2,guo2017calibration} addresses this by appending a score to each prediction, flagging low-confidence cases for manual review, relabeling, or rejection.

A range of methods have been proposed to estimate the confidence of a model's prediction. The simplest is Maximum Class Probability (MCP)~\cite{hendrycks2017a_baseline}, which uses the predicted softmax probability as the confidence score. However, deep neural networks are systematically overconfident~\cite{corbiere2019_confidence_main_model,corbiere2021confidence_confidnet_2,guo2017calibration}, causing MCP to assign high confidence even to incorrect predictions. Sampling- and ensemble-based approaches~\cite{blundell2015weight,gal2016dropout,lakshminarayanan2017simple} provide improved uncertainty estimates but require multiple forward passes or retraining. 
Post-hoc confidence estimation methods, such as ConfidNet~\cite{corbiere2019_confidence_main_model}, leave the classifier frozen and train a lightweight auxiliary confidence head. Their effectiveness largely depends on the confidence target used for supervision.
The original ConfidNet~\cite{corbiere2019_confidence_main_model} uses 
the True Class Probability (TCP), which reduces overconfidence by using the probability assigned to the ground-truth class. However, its target values for correct and incorrect predictions overlap, and this ambiguity increases with the number of classes. Its normalized variant, $TCP_n$~\cite{corbiere2019_confidence_main_model}, eliminates this overlap for correct predictions but assigns near-boundary errors confidence targets indistinguishable from correct predictions. 
These limitations arise from the design of the confidence target itself rather than from the post-hoc framework, motivating the search for improved target formulations.
Additionally, learning such confidence targets results in a highly imbalanced regression problem. Since the targets are computed from the predictions of an accurate base classifier, most samples receive high confidence targets, while only a small fraction corresponds to errors. Learning from such a skewed target distribution is therefore challenging.

In this work, we propose a modified target $\mathrm{TCP}_\alpha$, which augments the TCP\textsubscript{n} denominator with a penalty term active only on misclassified samples. The penalty term creates separation between mis-classified and correctly classified samples with a margin that is independent of the number of classes and widens monotonically with $\alpha$. 
The target imbalance, however, persists regardless of how the target is constructed, since it is inherited from the base classifier itself. 
We therefore study a range of training strategies to address it, and identify the configuration that works best through ablations.

We first validate the proposed method on r\=aga identification~\cite{Paramsingh2024explainabledeeplearninganalysis}, where the success–error imbalance is particularly severe. We then evaluate whether the same configuration transfers to frame-wise ornamentation detection~\cite{sumit_ornaments_TASLP_2025}, a task with different temporal granularity and class distribution. Finally, we study domain shift using a second r\=aga corpus~\cite{Saraga}.
While our experiments focus on Indian Art Music, the proposed framework can be readily applied to other deep classification problems requiring reliable failure prediction.

The contributions of this work are summarized as follows:

\begin{itemize}
  \item We propose TCP$_\alpha$, a novel confidence target for failure prediction that achieves complete separation between correct and incorrect predictions, independent of the number of classes.
  \item We provide theoretical guarantees establishing complete target separation and a margin independent of the number of classes.
  \item We identify the configuration that best handles the training imbalance on r\=aga identification through a systematic study of training strategies, and show that rejecting a very small fraction of samples yields large performance gains.
  \item We validate on frame-wise ornamentation detection, a task differing in granularity and imbalance, and achieve similar gains by rejection.
  \item We examine transfer under domain shift, and find that fine-tuning the confidence head on a few labeled target-domain samples restores well-separated confidence scores.
\end{itemize}

This work substantially extends our ICASSP workshop publication~\cite{sumit_confidence} by (1) providing theoretical guarantees for the proposed $\mathrm{TCP}_\alpha$ targets; (2) analyzing a range of training strategies for learning the highly imbalanced targets; (3) formulating the problem explicitly as failure prediction rather than confidence calibration; and (4) examining the effect of domain shift.

All the codes are available at: \url{https://tinyurl.com/tcp-alpha}

\section{Related Works}

Confidence and uncertainty estimation for deep neural networks have been studied extensively, with existing methods broadly falling into two categories. The first estimates uncertainty through the training process itself using Bayesian deep learning~\cite{blundell2015weight}, Monte Carlo dropout~\cite{gal2016dropout}, or deep ensembles~\cite{lakshminarayanan2017simple}. While these methods generally improve uncertainty estimation, they require specialized training procedures or multiple stochastic forward passes. The second category comprises \emph{post-hoc} confidence estimation methods, which operate on an already-trained classifier without modifying its decision boundary. Temperature scaling~\cite{guo2017calibration} calibrates the softmax probabilities to better match empirical accuracy, while Maximum Class Probability (MCP) uses the predicted softmax probability as a confidence score~\cite{hendrycks2017a_baseline}.
Other training-free scores include predictive entropy and the Top-2 Margin, which respectively use the full predictive distribution and the gap between the top two class probabilities~\cite{nguyen2022measure_entropy_top_2}, as well as the energy score, which derives a confidence measure directly from the pre-softmax logits~\cite{liu2020energy}. DOCTOR~\cite{granese2021neurips-doctor} similarly utilizes the full softmax vector for failure prediction without additional training.
ConfidNet~\cite{corbiere2019_confidence_main_model} instead trains a lightweight auxiliary confidence head over a frozen classifier using the True Class Probability (TCP) as the regression target, thereby reducing the overconfidence exhibited by MCP on incorrect predictions. Our work builds upon this post-hoc confidence estimation framework by proposing a new target formulation that addresses the limitations of existing TCP-based targets.

An important distinction in this area is between \emph{confidence calibration} and \emph{failure prediction}~\cite{zhu2023rethinking}.
Calibration seeks probabilistic correctness, ensuring that predicted confidence matches empirical accuracy, and is commonly evaluated using metrics such as the Expected Calibration Error (ECE)~\cite{guo2017calibration,naeini2015obtaining}. However, recent studies have shown that ECE can be insensitive to important calibration failures and does not always reflect the practical usefulness of confidence estimates~\cite{chidambaram2024how_ECE_flaw,Nixon_2019_CVPR_Workshops_ECE_flaws}. More importantly, good calibration does not necessarily imply good failure prediction. The work~\cite{zhu2023rethinking} demonstrate that several methods which improve calibration error or out-of-distribution detection can simultaneously reduce the separability between correct and incorrect predictions, which is precisely the opposite of the objective in failure prediction. Our goal is therefore not to estimate calibrated probabilities, but to learn a confidence score that reliably distinguishes correct predictions from incorrect ones, enabling unreliable predictions to be rejected, deferred for human review, or prioritized for relabeling. Accordingly, we evaluate using failure prediction metrics such as AUPR-Error, AUROC, and FPR@95\%TPR~\cite{corbiere2019_confidence_main_model} rather than calibration metrics like ECE.
Failure prediction is also closely related to selective prediction or confidence-based abstention, which studies using confidence or uncertainty to trade predictive risk against coverage~\cite{geifman2017selective,elyaniv2010foundations,franc2023optimal_reject,geifman2019selectivenet_reject}. Complementary to this line of work, we focus on learning a more reliable confidence score that can be used within such frameworks.


Learning confidence targets presents another challenge, as the confidence targets are highly imbalanced. Classical imbalance mitigation techniques such as class reweighting, oversampling, and focal loss~\cite{lin2017focal,cui2019class,chawla2002smote} are primarily designed for discrete class labels. Our setting is therefore closer to Deep Imbalanced Regression~\cite{yang2021delving}, which introduces Label Distribution Smoothing (LDS) and Feature Distribution Smoothing (FDS) for skewed continuous targets. However, unlike conventional regression problems, the imbalance in our setting is not an intrinsic property of the dataset but is induced by the high accuracy of the underlying classifier itself. We therefore investigate training strategies specifically tailored to this setting.

Within Music Information Retrieval (MIR), uncertainty and confidence estimation have received comparatively little attention. Existing studies have explored uncertainty in music genre classification~\cite{genre,ye22b_interspeech_genre_confidence}, musical improvisation~\cite{music_improv}, music emotion recognition~\cite{emotion}, and melody estimation~\cite{Ayush_ICASSP_evidential,melody_kavya}. In Indian Art Music, previous work has primarily focused on improving predictive performance for tasks such as \textit{r\=aga} identification~\cite{Paramsingh2024explainabledeeplearninganalysis,phononet_2019,raga_ontology,singh2025identification} and ornamentation detection~\cite{sumit_ornaments_TASLP_2025}, without estimating the reliability of model predictions.  
Our earlier work~\cite{sumit_confidence} took the first step toward closing this gap by introducing a modified TCP target for these two tasks.
The present work substantially extends that direction through a more rigorous theoretical treatment, improved learning strategies, and broader experimental validation.

\section{Preliminaries and Problem Formulation}
 
\subsection{Setup}
Consider a labeled dataset $\Dset=\{(x_i,\ystar_i)\}_{i=1}^{N}$ of $N$ i.i.d.\ samples, where $x_i\in\mathbb{R}^{d}$ denotes the input sample and $y_i^*\in\mathcal{Y}=\{1,\ldots,K\}$ is its ground-truth class.
Let $f_w$ denote a K-class classifier with parameters $w$ that induces a predictive distribution $\prob(Y\mid w,x)$ through a softmax layer, and let
\begin{equation}
  \yhat \;=\; \arg\max_{k\in\mathcal{Y}} \prob(Y=k\mid w,x).
  \label{eq:pred}
\end{equation}
Throughout this paper, we denote the probability assigned to the ground-truth class as $p^*$ and to the predicted class as $\hat p$:
\begin{equation}
  \pstar \;=\; \prob(Y=\ystar\mid w,x),
  \qquad
  \phat \;=\; \prob(Y=\yhat\mid w,x),
\end{equation}
respectively. 
We partition the dataset into the \emph{success} set $\Dpos=\{i:\yhat_i=\ystar_i\}$ ($f_w$ makes a correct prediction) and the \emph{error} set $\Dneg=\{i:\yhat_i\neq\ystar_i\}$ ($f_w$ makes an incorrect prediction), so that $\Dpos\cup\Dneg=\{1,\dots,N\}$.

Following the post-hoc confidence estimation framework of ConfidNet~\cite{corbiere2019_confidence_main_model}, the base classifier $f_w$ is trained once and then \emph{frozen}. As shown in Fig.~\ref{fig:Overview}, our objective is to attach to every prediction a scalar confidence score $\chat(x;\theta)\in[0,1]$, produced by an auxiliary head with parameters $\theta$ operating on the frozen backbone's penultimate representation $z(x)\in\R^{d_f}$. The head is trained by regression against a \emph{confidence target} $c^{*}(x,\ystar)$,
\begin{equation}
  \mathcal{L}_{\mathrm{conf}}(\theta;\Dset)
  \;=\;
  \frac{1}{N}\sum_{i=1}^{N}
  \bigl(\chat(x_i;\theta)-c^{*}(x_i,\ystar_i)\bigr)^{2}.
  \label{eq:conf-loss}
\end{equation}
The confidence target is a function of the ground-truth label and is therefore available only at training time. The task is to reproduce it at inference from $z(x)$ alone. 
Consequently, the effectiveness of this framework depends critically on the design of the target $c^*(x,y^*)$, which is the primary focus of this work.

\begin{figure}[h!]
    \centering
    \includegraphics[width=0.43\textwidth]{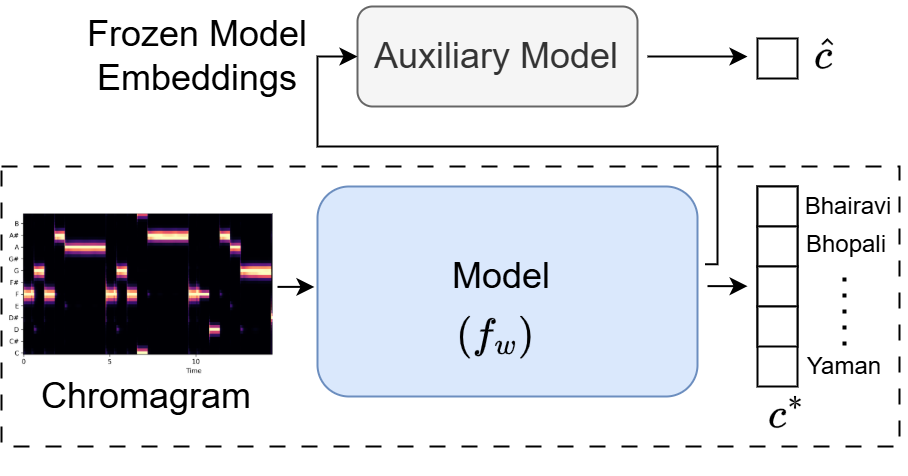}
    \caption{Overview of the post-hoc confidence estimation module. The
  backbone and classifier are frozen after base training; only
  the confidence head is trained, against the TCP$_\alpha$
  target derived from the frozen classifier's softmax output and the
  ground-truth label.}
    \label{fig:Overview}
\end{figure}

\subsection{Existing Targets and Their Failure Modes}
\label{sec:existing-targets}
We briefly review the existing target formulations and highlight the limitations that motivate our proposed target. 
\paragraph{Maximum Class Probability (MCP)} The simplest confidence score is $\MCP(x)=\phat$ \cite{hendrycks2017a_baseline}, which is just the probability of the predicted class for the base classifier. It suffers from the problem of models being over-confident. Since $\phat\ge 1/K$ always, MCP takes values in $[1/K,1]$ for both successes and for errors, which results in an overlap in confidence scores for both (Fig.~\ref{fig:mcp_TCP_proposed_comparison}).
 
\paragraph{True Class Probability (TCP)} 
ConfidNet~\cite{corbiere2019_confidence_main_model} instead learns $\TCP(x,\ystar)=\pstar$~\cite{corbiere2019_confidence_main_model}, which is the probability associated with the True class instead of the predicted class. 
Two guarantees follow immediately:
\begin{enumerate}
  \item if $\pstar>\tfrac12$ then $\yhat=\ystar$;
  \item if $\pstar<\tfrac1K$ then $\yhat\neq\ystar$.
\end{enumerate}
Neither is vacuous, but together they leave an ambiguity band: because an error forces $\pstar<\phat$ and hence $\pstar<\tfrac12$, the error targets occupy $[0,\tfrac12)$ while the success targets occupy $[\tfrac1K,1]$. Any target value in $[\tfrac1K,\tfrac12)$ is therefore attainable by a correct and by an incorrect prediction alike, and no threshold on TCP can separate the two. The width of this band grows with $K$.
 
\paragraph{Normalised TCP (TCP\textsubscript{n})} To eliminate the ambiguity on the success side, \cite{corbiere2019_confidence_main_model} normalises the TCP score by the predicted-class probability,
\begin{equation}
  \TCPn(x,\ystar) \;=\; \frac{\pstar}{\phat},
  \label{eq:tcpn}
\end{equation}
which pins every success to exactly $1$. The authors claim stronger theoretical guarantees compared to TCP, but this creates an additional problem for the error targets. Consider an error with $\pstar=m-\varepsilon$ and $\phat=m+\varepsilon$ for some $m$ and $\varepsilon>0$. Then $\TCPn\to 1$ as $\varepsilon\to 0^{+}$, even though the prediction is wrong. So, near-boundary errors are thus assigned targets indistinguishable from confident successes. 
 
\begin{figure}[htbp]
\includegraphics[width=0.9\columnwidth]{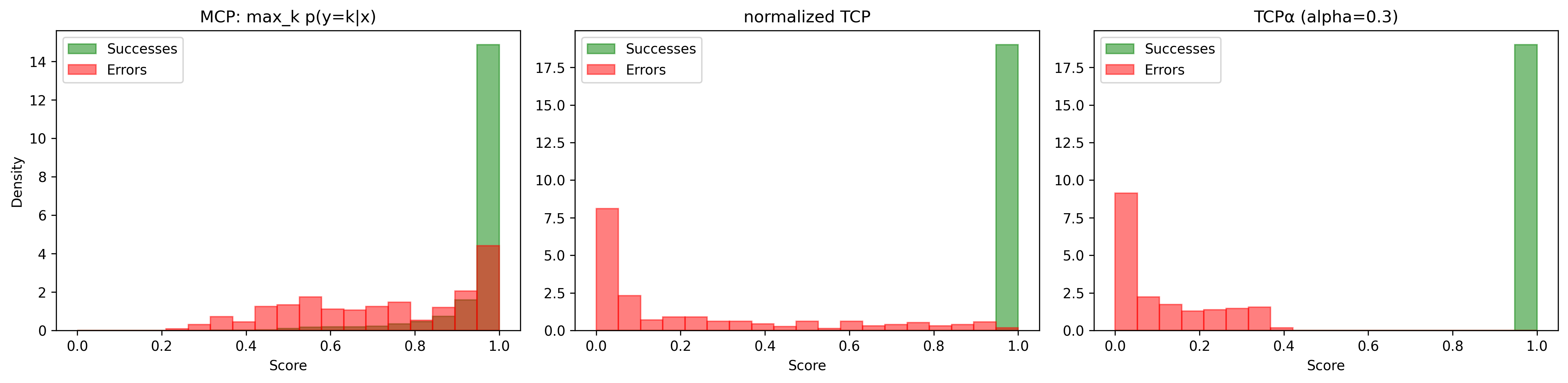}
\caption{Comparison of the target distributions induced by MCP, TCP$_n$, and the proposed TCP$_\alpha$ for PIM~\cite{Paramsingh2024explainabledeeplearninganalysis} dataset. The x-axis represents the targets and y-axis represents the density. TCP$_\alpha$ produces complete separation between success and error targets.}
\label{fig:mcp_TCP_proposed_comparison}
\end{figure}

\section{Proposed Method}
\subsection{The TCP\texorpdfstring{$_\alpha$}{alpha} Confidence Target}
\label{sec:target}
The limitations of TCP and TCP$_n$ arise because the corresponding targets cannot completely separate correct and incorrect predictions. We therefore introduce TCP$_\alpha$, which modifies TCP$_n$ through an additive penalty that is active only for misclassified samples.
We define: 
\begin{equation}
  \TCPa(x,\yhat,\ystar)
  \;=\;
  \frac{\pstar}
       {\phat \;+\; \ind{\yhat\neq\ystar}\,\bigl(\pstar+\alpha\bigr)},
  \label{eq:tcpa}
\end{equation}
where $\alpha\ge 0$ is a penalty hyperparameter. 
The proposed formulation preserves the desirable property that every correctly classified sample receives a target of one, while explicitly pushing all error targets away from the success region.
When $\yhat\neq\ystar$ the denominator is inflated by $\pstar+\alpha$, which strictly depresses the target below its TCP\textsubscript{n} value. Fig.~\ref{fig:mcp_TCP_proposed_comparison} shows a comparison of the MCP, TCP\textsubscript{n} and the proposed TCP\texorpdfstring{$_\alpha$}{alpha} targets.

It is worth noting that the $\pstar$ term inside the penalty is what removes the near-boundary pathology of TCP\textsubscript{n} \emph{even at $\alpha=0$}.
The constant $\alpha$ then provides continuous, monotone control over how far below the success targets the error targets are pushed.

\subsection{Theoretical Guarantees}
 
Let $\Tpos=\{\TCPa(x,\yhat,\ystar):\yhat=\ystar\}$ and $\Tneg=\{\TCPa(x,\yhat,\ystar):\yhat\neq\ystar\}$
denote the sets of achievable success and error targets respectively.
 
\begin{proposition}[Success targets are exactly one] \label{prop:success}
For any $x$ with $\yhat=\ystar$ and any $\alpha\ge 0$, $\TCPa(x,\yhat,\ystar)=1$. Hence $\Tpos=\{1\}$.
\end{proposition}
 
\begin{proof}
When $\yhat=\ystar$ the indicator in \eqref{eq:tcpa} vanishes and $\phat=\pstar$. So,
$\TCPa=\pstar/\pstar=1$.
\end{proof}
 
\begin{proposition}[Tight ceiling on error targets]
\label{prop:ceiling}
For any $x$ with $\yhat\neq\ystar$ and any $\alpha\ge 0$,
\begin{equation}
  \TCPa(x,\yhat,\ystar) \;<\; \frac{1}{2(1+\alpha)}.
  \label{eq:ceiling}
\end{equation}
The bound is tight: $\sup\Tneg = \tfrac{1}{2(1+\alpha)}$, approached in the limit $\phat\to\pstar$, $\pstar\to\tfrac12$.
\end{proposition}
 
\begin{proof}
For an error, \eqref{eq:tcpa} gives
$\TCPa=\pstar/(\phat+\pstar+\alpha)$. Strictness of the argmax gives
$\phat>\pstar$, hence
\begin{equation}
  \phat+\pstar+\alpha \;>\; 2\pstar+\alpha,
  \qquad\text{so}\qquad
  \TCPa \;<\; \frac{\pstar}{2\pstar+\alpha}.
  \label{eq:step1}
\end{equation}
Now, since $\phat>\pstar$ and $\phat+\pstar\le1$, we have
$2\pstar<\phat+\pstar\le1$, and hence $\pstar<\tfrac12$. Eq:~\eqref{eq:step1} becomes:

$$\TCPa<\tfrac{1/2}{1+\alpha}=\tfrac{1}{2(1+\alpha)}$$
 
For tightness, fix $\alpha\ge 0$ and take the sequence $\pstar_n=\tfrac12-\tfrac{1}{n}$, $\phat_n=\tfrac12-\tfrac{1}{2n}$, which satisfies $\phat_n>\pstar_n$ and $\phat_n+\pstar_n<1$ for all $n\ge 3$. 
Thus, each pair $(\pstar_n,\phat_n)$ corresponds to a valid error prediction.
Then,

\begin{equation}
\begin{aligned}
TCP_{\alpha n}
&=
\frac{\tfrac12-\tfrac1n}
     {\bigl(\tfrac12-\tfrac{1}{2n}\bigr)
      +\bigl(\tfrac12-\tfrac1n\bigr)+\alpha} \\
&=
\frac{\tfrac12-\tfrac1n}
     {1+\alpha-\tfrac{3}{2n}}
\xrightarrow{\;n\to\infty\;}
\frac{1}{2(1+\alpha)}.
\end{aligned}
\end{equation}

Thus, the target assigned to an incorrect prediction is strictly bounded by $\frac{1}{2(1+\alpha)}$, and this bound is tight
\end{proof}


\begin{theorem}[Complete target separation and separation margin]
\label{thm:separation}
For every $\alpha\ge 0$, $\Tpos\cap\Tneg=\emptyset$, and the separation
margin is
\begin{equation}
  \Delta(\alpha)
  \;\triangleq\;
  \inf_{v\in\Tpos} v \;-\; \sup_{v\in\Tneg} v
  \;=\; 1-\frac{1}{2(1+\alpha)} .
  \label{eq:margin}
\end{equation}
The margin is independent of $K$, satisfies $\Delta(0)=\tfrac12$, is
strictly increasing in $\alpha$, and $\Delta(\alpha)\to 1$ as
$\alpha\to\infty$.
\end{theorem}
 
\begin{proof}
By Proposition~\ref{prop:success}, $\inf\Tpos=1$. By Proposition~\ref{prop:ceiling},
$\sup\Tneg=\frac{1}{2(1+\alpha)}<1$. Hence,
$\Tpos\cap\Tneg=\emptyset$, and
\[
\Delta(\alpha)
=1-\frac{1}{2(1+\alpha)}.
\]
Finally,
\[
\frac{\mathrm{d}\Delta}{\mathrm{d}\alpha}
=\frac{1}{2(1+\alpha)^2}>0,
\]
so the margin increases strictly with $\alpha$, with
$\Delta(0)=\frac12$ and $\Delta(\alpha)\to1$ as
$\alpha\to\infty$.
\end{proof}

\begin{proposition}[Monotone penalty]
\label{prop:monotone}
For any fixed error sample with $\pstar>0$, $\TCPa$ is strictly decreasing in $\alpha$:
\begin{equation}
  \frac{\partial \TCPa}{\partial\alpha}
  \;=\; -\,\frac{\pstar}{(\phat+\pstar+\alpha)^{2}} \;<\; 0 .
\end{equation}
For success samples $\partial\TCPa/\partial\alpha=0$.
\end{proposition}
 
\begin{proof}
For errors, $C\triangleq\phat+\pstar>0$ is constant in $\alpha$ and $\TCPa=\pstar/(C+\alpha)$; differentiating gives the stated derivative, which is negative since $\pstar>0$. For successes Proposition~\ref{prop:success} gives $\TCPa\equiv 1$, whose derivative is zero.
\end{proof}
 
Proposition~\ref{prop:monotone} identifies $\alpha$ as a \emph{monotone margin-control parameter}: increasing it depresses every error target while leaving every success target pinned at $1$, so the two groups are driven apart uniformly and predictably.

\begin{remark}[Contrast with the baselines]
\label{rem:contrast}
For MCP the two target ranges coincide, for TCP the ranges overlap on $[\tfrac1K,\tfrac12)$, a band that \emph{widens} with the number of classes, so the guarantee degrades in exactly the many-class regime where confidence estimation is the most needed. For TCP\textsubscript{n} the ranges touch ($\sup\Tneg=\inf\Tpos=1$) and the margin is zero. Theorem~\ref{thm:separation} shows that TCP$_\alpha$ is the only member of the family with a strictly positive, $K$-independent margin, and it attains this already at $\alpha=0$. The role of $\alpha>0$ is therefore \emph{not} to create separation but to enlarge it. Table~\ref{tab:theory} summarises all these differences.
\end{remark}

 
 
\begin{table}[t]
\centering
\caption{Theoretical properties of confidence targets. $\Tpos,\Tneg$ are the achievable target ranges for successes and errors; $K$ is the number of classes; $\alpha\ge 0$ is the penalty of TCP$_\alpha$.}
\label{tab:theory}
\small
\begin{tabular}{lcccc}
\toprule
Criterion & $\Tpos$ & $\Tneg$ & Overlap & Margin $\Delta$ \\
\midrule
MCP  & $[\tfrac1K,1]$ & $[\tfrac1K,1]$ & always & $0$ \\
TCP  & $[\tfrac1K,1]$ & $[0,\tfrac12)$ & $[\tfrac1K,\tfrac12)$ & $0$ \\
TCP\textsubscript{n} & $\{1\}$ & $(0,1)$ & at boundary & $0$ \\
\textbf{TCP$_\alpha$ (ours)} & $\{1\}$ &
  $\bigl(0,\tfrac{1}{2(1+\alpha)}\bigr)$ & none &
  $1-\tfrac{1}{2(1+\alpha)}$ \\
\bottomrule
\end{tabular}
\end{table}

\subsection{Learning the Confidence Head}
\label{sec:method}

Since the base classifier is highly accurate, the target distribution induced by TCP$_\alpha$ is extremely skewed, with only a small fraction corresponding to errors. Although the target provides complete separation between correct and incorrect predictions, the resulting regression problem is dominated by the abundant high-confidence samples, making informative low-confidence predictions difficult to learn.

To address this challenge, we employ an imbalance-aware training strategy consisting of stratified mini-batch sampling and class-conditional loss weighting.
We enforce a fixed error-to-success ratio in every mini-batch during training. Let $r = |\mathcal{D}_-^{\text{batch}}| / |\mathcal{D}_+^{\text{batch}}|$ be the desired ratio. Since $|\mathcal{D}_-| \ll |\mathcal{D}_+|$, error samples are exhausted faster than successes under random sampling, so we maintain a circular buffer over $\mathcal{D}_-$. With $\pi^{(t)}$ the $t$-th random permutation of $\mathcal{D}_-$, the buffer for epoch $e$ concatenates successive permutations,
\begin{equation}
    \mathcal{B}_-(e) = \pi^{(1)} \,\|\, \pi^{(2)} \,\|\, \cdots,
\end{equation}
drawing fresh permutations until every mini-batch in the epoch is filled. Each error sample thus recurs several times before all successes are seen once, without repeating within a single permutation. We combine this with class-conditional weighting, which weights each sample by the inverse frequency of its base class in the error pool.



\section{Experimental Setup}\label{sec:setup}
 
\subsection{Datasets}\label{sec:datasets}

\subsubsection{R\=aga Identification}

The Prasar Bharati Indian Music (PIM) dataset~\cite{Paramsingh2024explainabledeeplearninganalysis} contains 191 hours of polyphonic Hindustani concert recordings annotated with r\=aga and tonic labels. Following~\cite{Paramsingh2024explainabledeeplearninganalysis}, we use the 12 most frequent r\=agas, segmented into non-overlapping 30-second audio chunks.
To evaluate robustness under domain shift, we use the Saraga-Hindustani dataset~\cite{Saraga}, an independently collected IAM corpus with different recording conditions and provenance. It contains recordings from 11 of the selected r\=agas, with substantially fewer samples per class than PIM. 

\subsubsection{Ornamentation Detection}

The R\=aga Ornamentation Detection corpus~\cite{sumit_ornaments_TASLP_2025} contains mono-channel expert Hindustani vocal recordings with frame-wise labels for seven ornament types recorded by 2 expert Hindustani Musicians. 
We use audio from Expert-2, segmented into 10-second clips with frame-wise ornamentation labels.

\begin{table}[t]
\centering
\caption{Dataset statistics and the induced imbalance seen by the confidence head. $|\Dneg|$ is the number of base-model errors available for training the confidence head.
K is the number of classes.
Base F1 shows performance of the base classifier on which the confidence head is being trained.}
\label{tab:datasets}
\small
\begin{tabular}{lccccc}
\toprule
Dataset & Unit & $K$ &
$|\Dpos|\!:\!|\Dneg|$ & Base F1 \\
\midrule
PIM      & 30\,s chunk & 12 & $\approx$8.1:1 & 0.89 \\
ROD    & frame       & 7 & $\approx$2.2:1        & 0.69 \\
Saraga   & 30\,s chunk & 11 & $\approx$3.2:1 & 0.76 \\
\bottomrule
\end{tabular}
\end{table}

\subsection{Base Models}

For \emph{r\=aga identification}, we follow~\cite{Paramsingh2024explainabledeeplearninganalysis} and use tonic-normalized chromagram features as input to a CNN--LSTM classifier trained with the cross-entropy loss on the PIM dataset. For \emph{ornamentation detection}, we use the Temporal Convolutional Network proposed in~\cite{sumit_ornaments_TASLP_2025}, which operates on chromagram features using periodic padding and dilated convolutions.

Both classifiers are trained to convergence and subsequently frozen. The proposed confidence head is then attached to the penultimate feature representation while leaving the backbone and the original classification layer unchanged. For the domain-shift experiments on Saraga dataset, we directly evaluate the confidence head using the classifier trained on PIM, followed by fine-tuning on a few labeled  Saraga samples.

Table~\ref{tab:datasets} summarizes the datasets, the performance of the corresponding base classifiers, and the resulting success-to-error ratios. The final column highlights the central challenge addressed in this work: a classifier with a macro-F1 of approximately $0.89$ produces only about one error for every eight correct predictions, resulting in a highly imbalanced confidence learning problem.

Unless otherwise stated, all ablation studies are performed on the PIM dataset. For the Saraga and ornamentation datasets, we report only the primary baselines and the final selected training strategy.

\subsection{Baselines}
\label{sec:baselines}
We compare the proposed method against two categories of confidence estimators. The first comprises training-free scores computed directly from the frozen classifier, while the second consists of learned confidence heads that differ only in the target used for supervision.

\paragraph{MCP} MCP, described in Section~\ref{sec:existing-targets}, serves as the simplest confidence baseline.

\paragraph{Predictive Entropy}
Predictive entropy~\cite{nguyen2022measure_entropy_top_2} measures the uncertainty of the complete predictive distribution,
\begin{equation}
  \PE(x) \;=\; -\sum_{k=1}^{K} \prob(Y=k\mid w,x)\, \log \prob(Y=k\mid w,x).
  \label{eq:entropy}
\end{equation}
Lower entropy indicates higher confidence; we therefore use $-\PE(x)$ as the confidence score. Unlike MCP, it accounts for the entire predictive distribution rather than only the predicted class.

\paragraph{Top-2 Margin}
The Top-2 Margin~\cite{nguyen2022measure_entropy_top_2} measures the separation between the predicted class and its closest competitor,
\begin{equation}
  \Margin(x) \;=\; \phat \;-\;
    \max_{k \neq \yhat} \prob(Y=k\mid w,x).
  \label{eq:top_two_margin}
\end{equation}
Larger margins correspond to higher confidence and indicate predictions that lie farther from the decision boundary.

\paragraph{Energy Score}
The energy score~\cite{liu2020energy} is computed directly from the pre-softmax logits,
\begin{equation}
  \Energy(x) \;=\; -T \log \sum_{k=1}^{K} \exp\bigl(g_k(x)/T\bigr),
  \label{eq:energy}
\end{equation}
where $T=1$ in our experiments. Higher confidence corresponds to lower energy, so we use $-\Energy(x)$ as the confidence score.

\paragraph{MC Dropout}
It approximates Bayesian model averaging by performing $T_{\mathrm{mc}}$ stochastic forward passes with dropout enabled inference~\cite{gal2016dropout}, yielding sampled weights $w_t$ and the averaged predictive distribution
\begin{equation}
  \bar{p}_k(x) \;=\; \frac{1}{T_{\mathrm{mc}}}
    \sum_{t=1}^{T_{\mathrm{mc}}} \prob(Y=k\mid w_t,x),
  \label{eq:mcdropout}
\end{equation}
where $T_{\rm mc}=50$ in our experiments. The confidence score is taken as the negative predictive entropy of the averaged predictive distribution.

\paragraph{\emph{TCP} and \emph{TCP\textsubscript{n}}}
Here we train an auxiliary confidence head, implemented as a two-layer MLP with a sigmoid output using the target formulations described in Section~\ref{sec:existing-targets}. To ensure a fair comparison, both are trained using the same confidence head architecture and optimization strategy as the proposed TCP$_\alpha$ model.

\subsection{Ablation Study Design}

In addition to the proposed imbalance-aware training strategy described in Section~\ref{sec:method}, we evaluate several alternative approaches for learning the confidence head.
under the highly imbalanced TCP$_\alpha$ target distribution. 

\subsubsection{Regression with Imbalance Correction}
As a baseline, we first train the confidence head using the mean squared error loss of Eq.~(\ref{eq:conf-loss}) without any imbalance handling. We then compare it with several regression strategies designed to alleviate target imbalance. Label Distribution Smoothing (LDS)~\cite{yang2021delving} smooths the empirical target distribution and reweights samples according to the smoothed density, while its feature-space counterpart, Feature Distribution Smoothing (FDS)~\cite{yang2021delving}, performs smoothing in the learned feature representation. We also evaluate simple error upsampling, which creates duplicate copies of error samples during training. 

\subsubsection{Binary Discrimination}
Instead of regressing continuous confidence targets, we also investigate treating confidence estimation as a binary discrimination problem by assigning each sample a success/error label $\ell_i$ as:
\begin{equation}
  \ell_i = \ind{i \in \Dpos} =
  \begin{cases} 1, & \yhat_i = y_i^*,\\ 0, & \yhat_i \neq y_i^*, \end{cases}
\end{equation}
and training the confidence head using binary cross-entropy.
We additionally evaluate focal loss~\cite{lin2017focal}, which emphasizes difficult minority samples by down-weighting well-classified examples.
During inference, we take its predicted scalar output directly as the confidence score. 

\subsubsection{Per-Class Models}
Finally, we investigate whether replacing a single global confidence head with class-specific models improves performance. Two binary discrimination variants are considered. The first one is trained using only the successes and errors of the corresponding class given by:
\begin{equation}
  \Dset^{(k)}_{\text{cls}}
  = \underbrace{\{(x_i,1):\ystar_i=k,\;\yhat_i=k\}}_{\text{class-$k$ successes}}
  \;\cup\;
  \underbrace{\{(x_i,0):\ystar_i=k,\;\yhat_i\neq k\}}_{\text{class-$k$ errors}}.
  \label{eq:perclass-own}
\end{equation}
Hence, we train $k$ separate models. In another configuration, we also incorporate \emph{confusing negatives}, i.e., samples from other classes incorrectly predicted as the target class by the base classifier, denoted by $\mathcal{H}_k$:
\begin{equation}
  \mathcal{H}_k = \{\,(x_i,0):\, i\in\Dneg,\ \yhat_i = k\,\}.
  \label{eq:confusing}
\end{equation}
So here, we work with:
$\Dset^{(k)} = \Dset^{(k)}_{\text{cls}} \cup \mathcal{H}_k$, giving higher weight to the confusing $\mathcal{H}_k$ in the regression loss. 
Additionally, we evaluate per-class regression models trained on $\Dset^{(k)}$ using the proposed TCP$_\alpha$ targets in place of binary targets $\ell_i$. During inference, each sample is routed to the confidence head corresponding to its predicted class.

\subsection{Evaluation Metrics} \label{sec:metrics}
We evaluate using the standard failure prediction metrics: AUPR-Error (AUPR-E), FPR@95\%TPR, AUROC, and AUPR-Success (AUPR-S). 
Given the severe class imbalance, AUPR-E and FPR@95\%TPR are the primary metrics because they evaluate performance on the rare error class. AUROC measures overall separability, while AUPR-S is reported only for completeness.

\section{Results and Discussion}
\label{sec:results}

\begin{table}[t]
\centering
\caption{Performance on r\=aga identification (PIM). All values are percentages, best in bold. (R) denotes the target trained with vanilla regression, (P) denotes training with our proposed training strategy (Sec:~\ref{sec:method}). T2M: Top-2-Margin score, MC-D: MC Dropout}
\label{tab:main_all_results}
\small
\begin{tabular}{lcccc}
\toprule
Method & FPR@95 $\downarrow$ & AUPR-E $\uparrow$ & AUPR-S $\uparrow$ & AUROC $\uparrow$ \\
\midrule

MCP                & 42.25 & 50.59 & 96.94 & 83.60 \\
Entropy & 41.00 & 53.70 & 97.10 & 84.49 \\
T2M       & 43.00 & 44.25 & 96.70 & 82.53 \\
Energy       & 37.50 & 56.87 & 97.47 & 86.25 \\
MC-D         & 40.75 & 52.37 & 97.13 & 84.70 \\

\midrule
TCP (R)                          & 46.02    & 25.35    & 98.85    & 85.11    \\
$\mathrm{TCP}_{n}$ (R)           & 61.74 & 32.64 & 98.42 & 81.74 \\
$\mathrm{TCP}_{\alpha}$ (R) & 53.68 & 74.30 & 98.89 & 91.26 \\
TCP (P)                          & 12.24 & 83.13 & 99.76 & 97.37 \\
$\mathrm{TCP}_{n}$ (P)           & 41.24 & 76.01 & 99.20 & 92.74 \\
$\mathrm{TCP}_{\alpha}$ (P) & \textbf{1.60} & \textbf{95.96} & \textbf{99.95} & \textbf{99.46} \\



\bottomrule
\end{tabular}
\end{table}

\subsection{Comparison with Existing Baselines}

Table~\ref{tab:main_all_results} compares the proposed method against both training-free confidence scores and learned confidence estimators on the PIM dataset. Among the training-free methods, all confidence scores perform similarly, which is expected since MCP, predictive entropy, Top-2 Margin, and the energy score are deterministic functions of the same classifier output, while MC Dropout estimates uncertainty from multiple stochastic forward passes of the same model. Consequently, they induce similar rankings over correct and incorrect predictions. The energy score is the strongest training-free baseline, achieving an AUPR-E of 56.87 and an AUROC of 86.25, whereas MC Dropout offers no noticeable improvement despite requiring 50 forward passes.

Among the learned confidence models, TCP$_\alpha$ consistently outperforms both TCP and TCP$_n$ under identical training settings. Even with vanilla regression, TCP$_\alpha$ substantially improves failure prediction, confirming that the proposed target formulation itself is more informative. Applying the proposed imbalance-aware training strategy further improves all three learned targets, with TCP$_\alpha$ achieving the best overall performance (AUPR-E = 95.96, AUROC = 99.46 and FPR@95\%TPR = 1.60). TCP$_n$ consistently performs the worst among the learned targets because its formulation assigns high target values to many near-boundary errors, making them difficult for the confidence head to distinguish from correct predictions. The rejection curves in Fig.~\ref{fig:sample_rejection_curve} demonstrate the practical impact of this improvement: rejecting only the least-confident $8\%$ of predictions increases the retained macro-F1 from 0.89 to approximately 0.98. The confidence density plots (Fig.~\ref{fig:best_model_density_plot}) further show that TCP$_\alpha$ produces substantially cleaner separation between correct and incorrect predictions than the existing targets.

\begin{figure}[htbp]
\includegraphics[width=0.9\columnwidth]{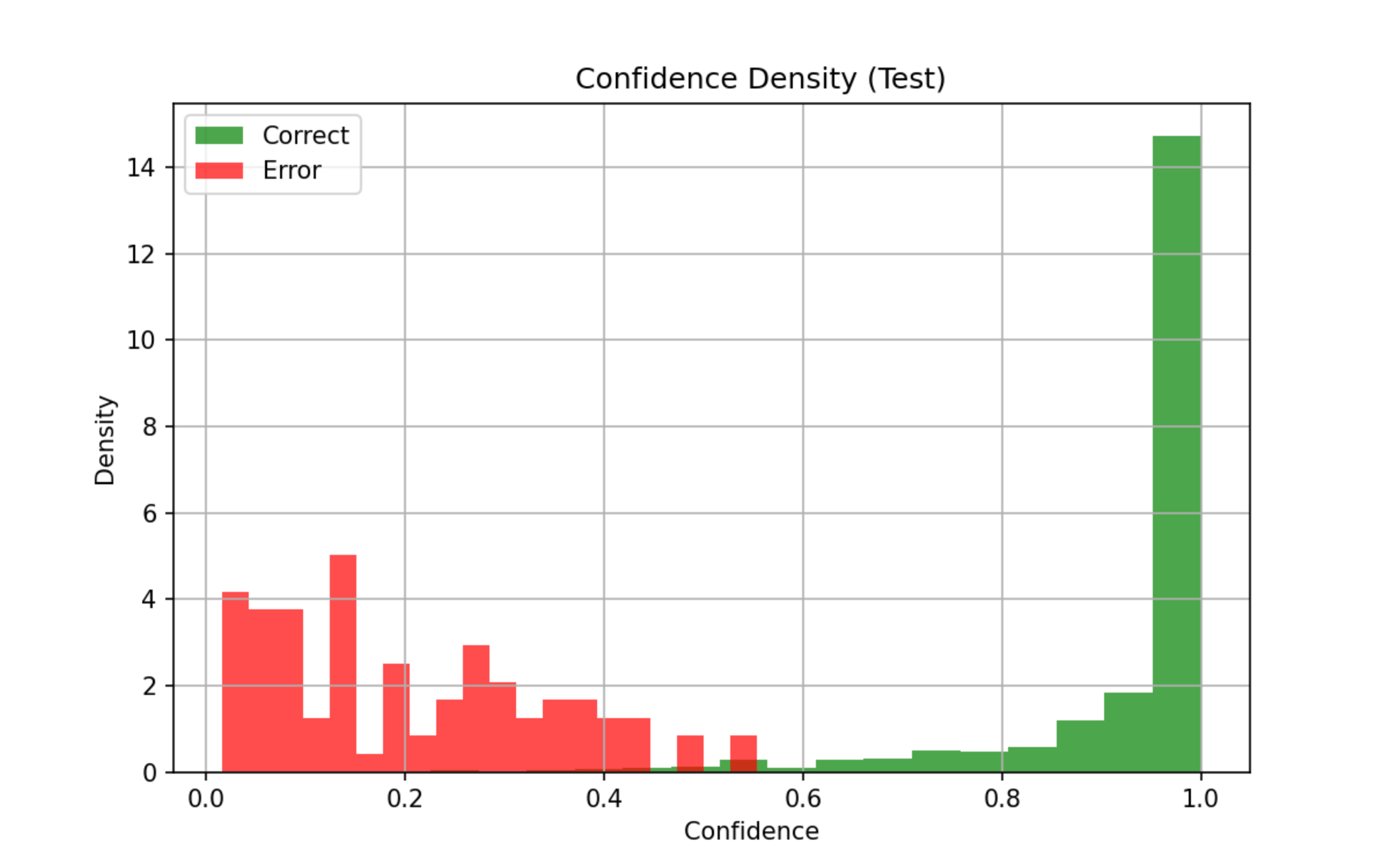}
\caption{Distribution of predicted confidence scores for $\mathrm{TCP}_{\alpha}$ on the PIM test set. The proposed target produces clear separation between correct and incorrect predictions.}
\label{fig:best_model_density_plot}
\end{figure}

\begin{figure}[htbp]
\includegraphics[width=0.9\columnwidth]{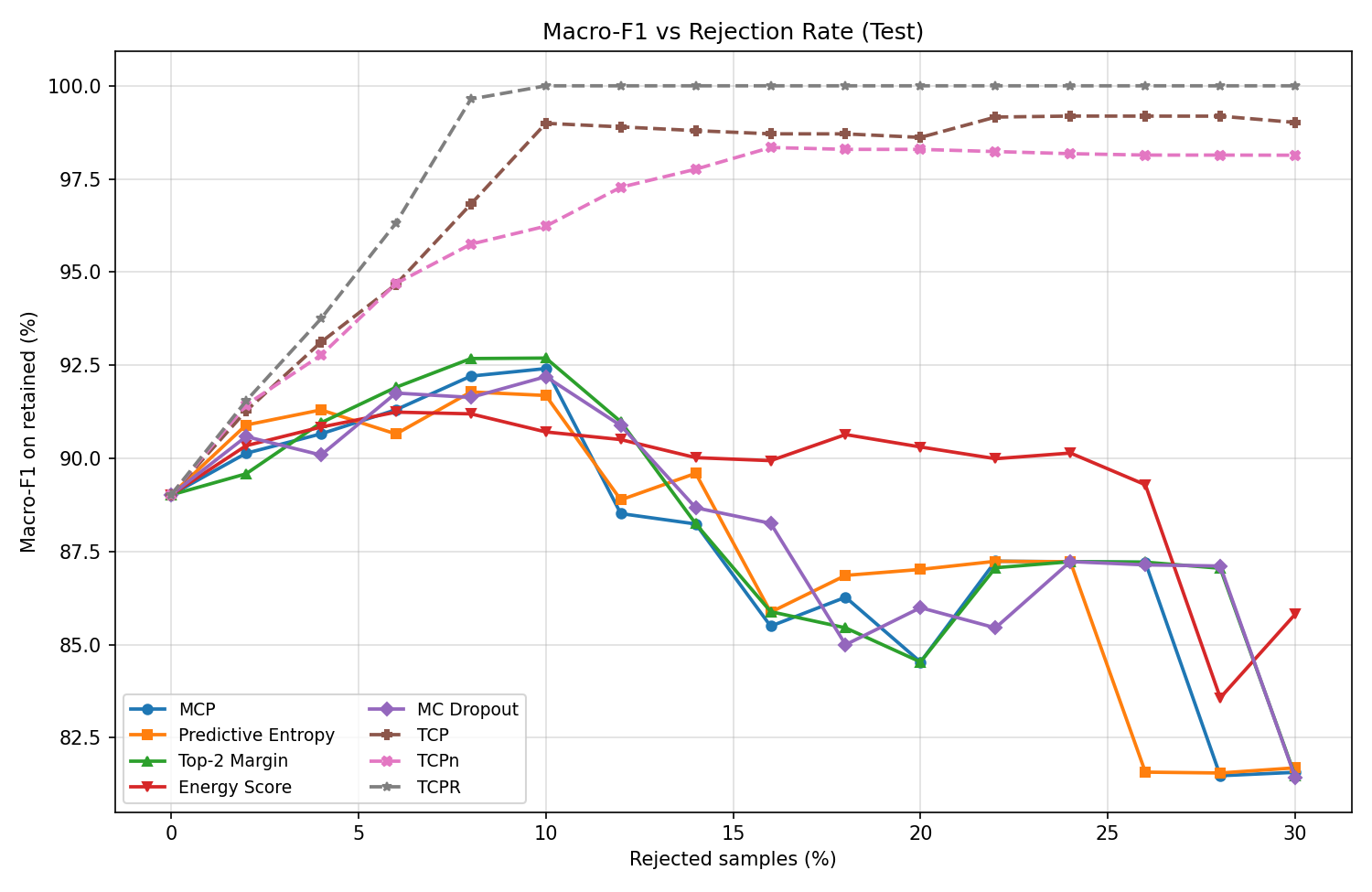}
\caption{Macro-F1 of the retained predictions as increasingly low-confidence samples are rejected using confidence scores from baselines and proposed method.}
\label{fig:sample_rejection_curve}
\end{figure}


\subsection{Ablation Studies}
\subsubsection{Training Strategy}
Table~\ref{tab:proposed_ablation} compares different strategies for learning the TCP$_\alpha$ confidence head under severe target imbalance. All methods use the same confidence head and TCP$_\alpha$ target, differing only in the optimization strategy.

Among the regression-based methods, LDS substantially improves over vanilla regression, but the confidence distributions remain considerably more diffuse than those produced by the proposed method. Instead of concentrating the correctly classified samples near confidence one, LDS spreads them over a broad high-confidence region while simultaneously shifting many error samples toward lower confidence. Consequently, the overlap between the two populations is reduced but not eliminated, with a noticeable number of successes still assigned relatively low confidence. The proposed training strategy, in contrast, produces a much sharper separation between the two distributions, resulting in substantially better failure prediction performance.
FDS provides no benefit and error upsampling performs poorly, likely because repeatedly sampling the same small set of error examples increases overfitting without improving diversity. 
Treating confidence estimation as a binary discrimination problem using binary cross-entropy or focal loss is consistently inferior to confidence regression, indicating that the continuous TCP$_\alpha$ target contains richer supervision than hard success-error labels. 

Treating confidence estimation as a binary discrimination problem using binary cross-entropy or focal loss is consistently inferior to confidence regression, indicating that the continuous TCP$_\alpha$ target contains richer supervision than hard success-error labels. Since binary models learn only two target values, even a few errors predicted as successes or a few successes assigned low confidence significantly degrade failure prediction metrics, as observed in their density plots. Among the per-class models, incorporating the confusing negatives $\mathcal{H}_k$ improves performance over the basic per-class discriminator, confirming that hard negatives are beneficial. However, neither the per-class classifiers nor the per-class regressors outperform a single global confidence model. For the regressors in particular, the small number of error samples available within each class provides insufficient supervision to learn reliable class-specific confidence functions. Such approaches may become more effective on datasets with substantially larger numbers of error samples per class.
The proposed imbalance-aware training strategy achieves the best performance by a large margin, improving AUPR-E from 74.30 to 95.96 while reducing FPR@95\%TPR from 53.68 to only 1.60. The corresponding confidence distributions (Fig.~\ref{fig:best_model_density_plot}) show a clear separation between correct and incorrect predictions, and we use this configuration for all subsequent experiments.

\begin{table}[t]
\centering
\caption{Ablation studies for training $\mathrm{TCP}_\alpha$ target on PIM. All rows use the same target and the same confidence head, and differ only in how the imbalance is handled. EU: Error Upsampling, $\mathrm{TCP}_\alpha$-P: proposed (Sec~\ref{sec:method}) }
\label{tab:proposed_ablation}
\small
\begin{tabular}{lcccc}
\toprule
Strategy & FPR@95 $\downarrow$ & AUPR-E $\uparrow$ &
AUPR-S $\uparrow$ & AUROC $\uparrow$ \\
\midrule
reg       & 53.68 & 74.30 & 98.89 & 91.26 \\
FDS                     & 54.71 & 65.92 & 98.89 & 90.07 \\
EU        & 66.14 & 29.91 & 98.34 & 81.67 \\
LDS                     & 21.98 & 85.70 & 98.94 & 94.42 \\
$\mathrm{TCP}_\alpha$ (P)  & \textbf{1.60}  & \textbf{95.96} & \textbf{99.95} & \textbf{99.46} \\
Binary & 76.08 & 25.92 & 97.13 & 76.75 \\
Focal  & 66.67  & 24.28 & 97.62 & 78.52 \\
P-C  & 68.79  & 50.61 & 97.95 & 82.10 \\ 
P-C-W  & 63.65  & 45.92 & 98.52 & 86.17 \\ 
K-reg         & 77.26    & 29.52    & 96.99    & 79.14    \\ 
\bottomrule
\end{tabular}
\end{table}

\begin{table}[t]
\centering
\caption{Effect of the success-to-error ratio $r$ in each training batch on PIM. Batch size is fixed at 64.}
\label{tab:success_error_ratio}
\small
\begin{tabular}{lccc}
\toprule
r & FPR@95 $\downarrow$ &
AUPR-E $\uparrow$ & AUROC $\uparrow$ \\
\midrule
5.4 & 39.36 & 91.09 & 95.08 \\
3.3 & 10.62 & \textbf{97.41} & 97.50 \\
2.2 & \textbf{1.60} & 95.96 & \textbf{99.46} \\
1.6 & 7.53 & 89.21 & 98.14 \\
1.1 & 7.63 & 82.19 & 98.25 \\
0.6 & 10.17 & 73.69 & 97.31 \\
0.3 & 13.28 & 52.81 & 95.19 \\
\bottomrule
\end{tabular}
\end{table}

\subsubsection{Effect of the Success:Error Batch Ratio}
Table~\ref{tab:success_error_ratio} studies the influence of the success-to-error ratio $r$ within each training mini-batch while keeping all other settings fixed. Very small numbers of error samples bias learning toward the dominant success class, leading to high FPR@95\%TPR despite good average ranking. 
Conversely, increasing the proportion of error samples gradually degrades AUPR-E. For a ratio of 0.3:1, we observe that the error samples are concentrated below a confidence value of 0.5, but the confidence assigned to correctly classified samples becomes increasingly dispersed in the whole range of (0,1) instead of remaining concentrated near one. As a result, the overlap between the two populations increases, making them less distinguishable. A ratio of approximately $2.2{:}1$ (20 error samples in a batch of 64) provides the best trade-off and is therefore used in all subsequent experiments.

\subsubsection{Effect of the Penalty Parameter $\alpha$}

Table~\ref{tab:alpha_ablation} evaluates the influence of the penalty parameter $\alpha$. Although Theorem~\ref{thm:separation} guarantees complete target separation even for $\alpha=0$, introducing a positive penalty provides an additional margin between the success and error targets, making the regression problem easier to learn. Accordingly, all positive values of $\alpha$ outperform $\alpha=0$, with the best performance obtained at $\alpha=1/K$. Beyond this point, the performance varies only slightly.
We use $\alpha=1/K$ throughout the experiments.


\begin{table}[t]
\centering
\caption{Effect of the penalty $\alpha$ on $TCP_\alpha$ performed on PIM dataset. 
}
\label{tab:alpha_ablation}
\small
\begin{tabular}{lcccc}
\toprule
$\alpha$  & FPR@95 $\downarrow$ &
AUPR-E $\uparrow$ & AUROC $\uparrow$ \\
\midrule
$0$       & 11.02 & 88.10 & 97.42 \\
$1/2K$    & 3.48 & 91.15 & 97.29 \\
$1/K$     & 1.60 & 95.96 & 99.46 \\
$2/K$     & 6.31 & 90.20 & 96.90 \\
$0.5$     & 8.10 & 87.49 & 96.77 \\
$1$       & 13.09 & 89.33 & 98.02 \\
$5$     & 8.38 &90.40 & 96.71 \\ %
\bottomrule
\end{tabular}
\end{table}



\begin{table}[t]
\centering
\caption{Performance on Saraga dataset. The base classifier is the PIM model in all rows and is never retrained. reg represents vanilla regression training. $\mathrm{TCP}_\alpha$ ($P_s$): best model trained on PIM fine-tuned on Saraga 
using s\% labeled training samples, test set being fixed for all.}
\label{tab:saraga}
\small
\begin{tabular}{lcccc}
\toprule
Model & FPR@95 $\downarrow$ & AUPR-E $\uparrow$ &
AUPR-S $\uparrow$ & AUROC $\uparrow$ \\
\midrule
MCP & 73.38 & 60.39 & 89.85 & 80.61 \\

reg & 52.83 & 80.14 & 92.05 & 88.56 \\
$\mathrm{TCP}_\alpha$ $(P_0)$  & 87.5 & 40.16 & 78.15 & 62.61 \\
$\mathrm{TCP}_\alpha$ $(P_{5})$ & 5.0 & 97.90 & 99.50 & 98.21 \\
$\mathrm{TCP}_\alpha$ $(P_{10})$ & 3.12 & 97.28 & 99.46 & 98.78 \\
$\mathrm{TCP}_\alpha$ $(P_{20})$ & 2.50 & 97.92 & 99.37 & 98.60 \\
$\mathrm{TCP}_\alpha$ $(P_{40})$ & 2.25 & 98.77 & 99.70 & 99.36 \\
\bottomrule
\end{tabular}
\end{table}

\subsection{Transfer to Domain Shift}
Table~\ref{tab:saraga} evaluates the proposed confidence model under domain shift using the Saraga dataset. The base r\=aga classifier trained on PIM is kept fixed throughout, and only the confidence head is adapted. Directly applying the confidence head learned on PIM to Saraga ($\mathrm{TCP}_\alpha$$(P_0)$) performs poorly, even below MCP, indicating that the learned confidence distribution is dataset-specific and does not transfer reliably across recording conditions.

Encouragingly, fine-tuning the confidence head using only $5\%$ labeled samples from the target domain restores almost all of the performance, improving AUPR-E from 40.16 to 97.90 while reducing FPR@95\%TPR from 87.5 to 5.0. Increasing the adaptation set beyond $5\%$ yields only marginal improvements, suggesting that a small amount of labeled target-domain data is sufficient to adapt the confidence model to the new domain.

\begin{table}[t]
\centering
\caption{Failure prediction on ornamentation detection. The
configuration is the one selected on r\=aga identification; no
additional search was performed.}
\label{tab:ornamentation_rod}
\small
\begin{tabular}{lcccc}
\toprule
Method & FPR@95 $\downarrow$ & AUPR-E $\uparrow$ &
AUPR-S $\uparrow$ & AUROC $\uparrow$ \\
\midrule
MCP  & 68.66 & 43.34 & 91.98 & 77.14 \\
TCP  & 33.46 & 76.62 & 97.60 & 92.31 \\
$\mathrm{TCP_n}$  & 47.88 & 71.21 & 96.77 & 89.22    \\
$\mathrm{TCP}_\alpha$ & \textbf{19.13}& \textbf{86.01}&\textbf{98.70}&\textbf{95.80} \\
Binary & 35.30 & 81.58 & 97.74 & 92.90 \\
Focal & 35.10 & 81.41 & 97.91 & 92.96\\
\bottomrule
\end{tabular}
\end{table}

\begin{figure}[htbp]
\includegraphics[width=0.9\columnwidth]{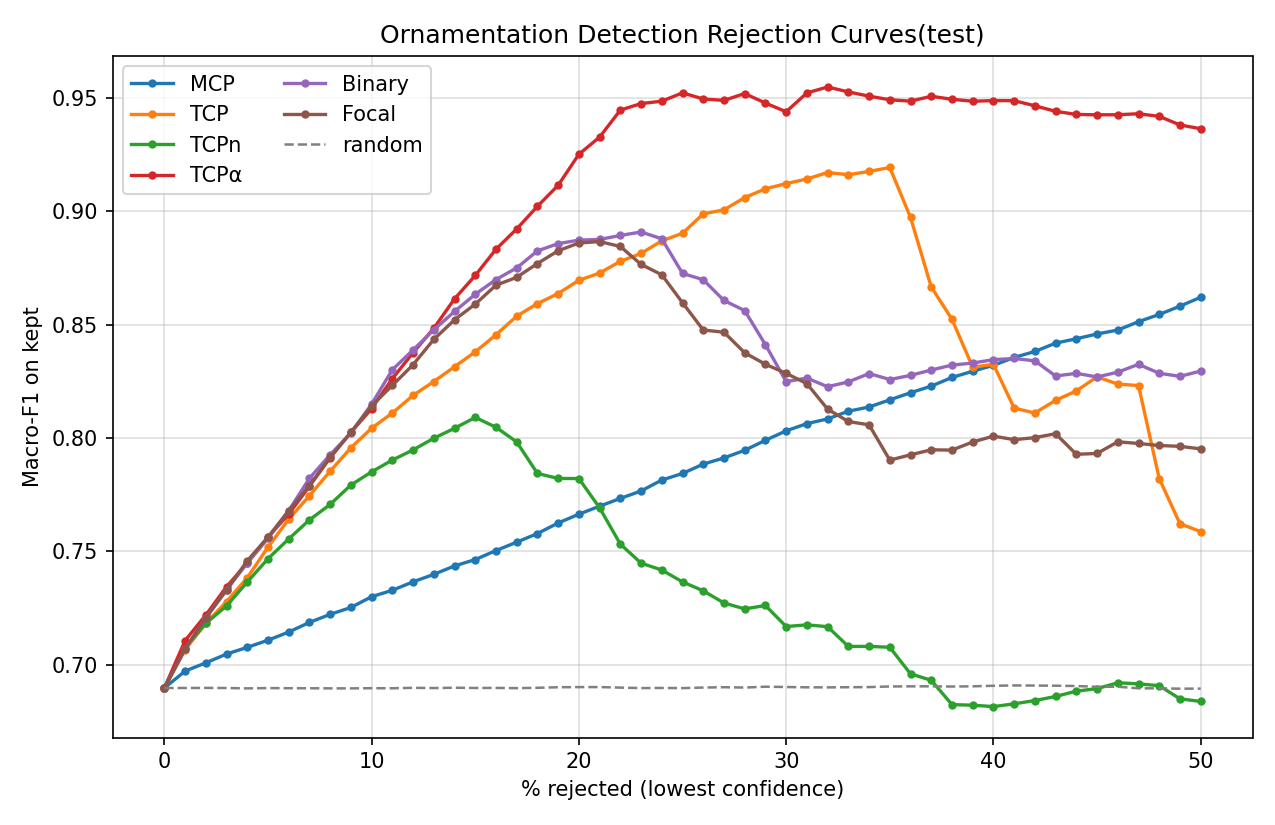}
\caption{Rejection Curves for Ornamentation detection task}
\label{fig:Ornament_rejection_curve}
\end{figure}
\begin{figure}[htbp]
\includegraphics[width=0.9\columnwidth]{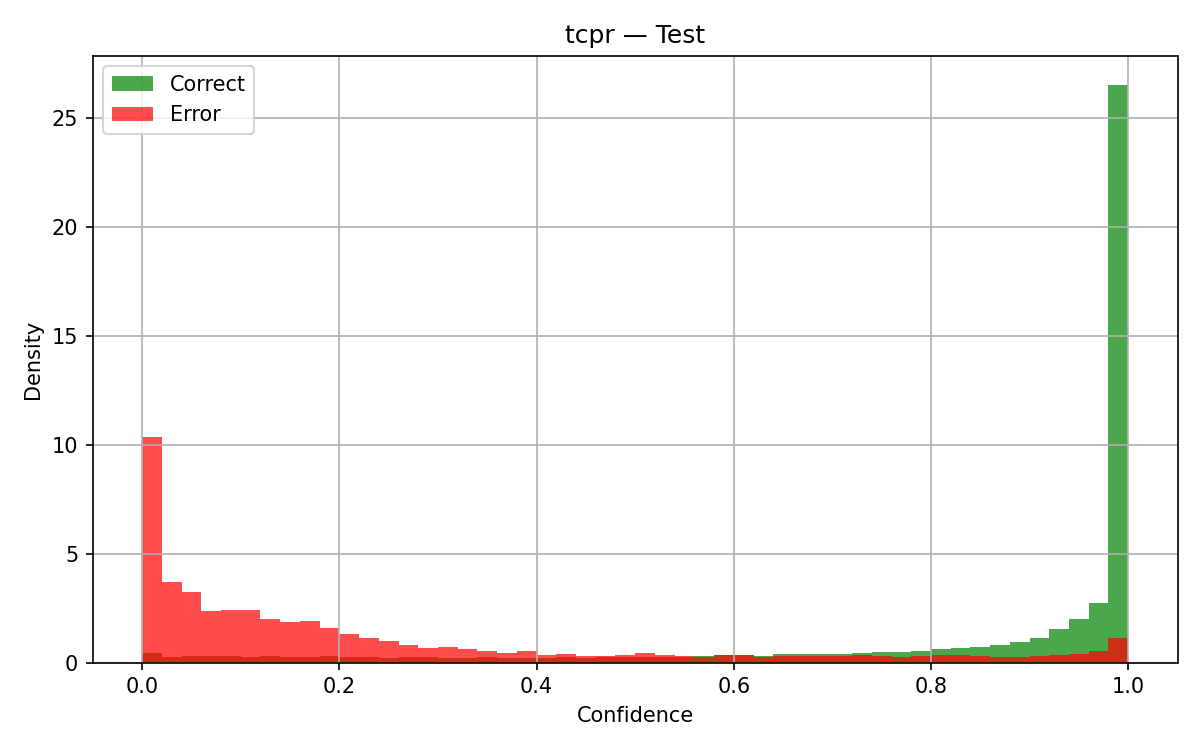}
\caption{Density plot for the best performing model for Ornamentation detection}
\label{fig:best_saraga_model_density_plot}
\end{figure}

\subsection{Ornamentation Detection}
Table~\ref{tab:ornamentation_rod} evaluates the proposed method on frame-wise ornamentation detection using the same configuration selected for r\=aga identification. TCP$_\alpha$ again achieves the best failure prediction performance across all metrics, improving AUPR-E to 86.01 while reducing FPR@95\%TPR to 19.13. Similar to the r\=aga identification task, the binary discrimination models improve upon MCP but remain inferior to confidence regression, confirming that the continuous TCP$_\alpha$ target provides richer supervision than hard success-error labels.

The rejection curves in Fig.~\ref{fig:Ornament_rejection_curve} further demonstrate the practical benefit of the proposed confidence estimator. Rejecting approximately $20\%$ of the least-confident frames increases the macro-F1 from 0.69 to approximately 0.95, after which the curve saturates, indicating that most erroneous predictions have already been removed. In contrast, TCP and the binary models peak earlier and then decline, reflecting their weaker separation between correct and incorrect predictions.

\section{Conclusion and Future Work}
\label{sec:conclusion}

In this work, we propose TCP$_\alpha$, a novel confidence target for post-hoc failure prediction that provides complete separation between correct and incorrect predictions with a margin independent of the number of classes. We study several training strategies for the resulting imbalanced regression problem and identify an imbalance-aware configuration that transfers across tasks. Experiments on r\=aga identification and ornamentation detection demonstrate that the proposed approach consistently outperforms existing confidence targets and commonly used confidence scores. Rejecting predictions using the proposed confidence scores yields substantial improvements in the retained predictions, showing that a small amount of expert review directed by TCP$_\alpha$ is sufficient to make the remaining predictions highly reliable. We also observe that only a small amount of labeled target-domain data is sufficient to adapt the confidence head to a new domain.

A promising direction for future work is to develop class-conditional confidence estimation that not only identifies unreliable predictions but also indicates plausible alternative classes. This is particularly useful for music analysis, where confusing two closely related r\=agas is fundamentally different from confusing unrelated ones. Such information can provide more meaningful feedback in downstream applications such as music analysis and pedagogy.

\bibliographystyle{IEEEtran}
\bibliography{bib}

\end{document}